\documentclass{article}
\usepackage[latin1]{inputenc}
\usepackage[english]{babel}
\usepackage{dsfont}
\usepackage[cyr]{aeguill}
\usepackage{amsmath,amssymb}
\usepackage{amsthm}
\usepackage{eurosym}
\usepackage{multicol}
\usepackage{color}
\usepackage{fancyhdr}
\usepackage[pdftex]{graphicx}
\usepackage[pdftex]{hyperref}
\usepackage{verbatim}
\usepackage{setspace}
\usepackage{stmaryrd}

\newcommand{\proba}{\mathbb{P}} 
\newcommand{\indic}{\mathds{1}} 
\newcommand{\E}{\mathbb{E}} 

\theoremstyle{plain}
\newtheorem{thm}{Theorem}
\newtheorem{pro}{Proposition}

\newtheorem{lem}{Lemma}

\providecommand{\keywords}[1]{\textbf{\textit{Keywords --}} #1}

\title{A statistical test of market efficiency \\ based on information theory}

\author{Xavier Brouty$^{\text{a}}$, Matthieu Garcin$^{\text{b,}}$\thanks{Corresponding author: matthieu.garcin@m4x.org. \newline $^{\text{a}}$ ESILV, 92916 Paris La Défense, France. \newline $^{\text{b}}$ Léonard de Vinci Pôle Universitaire, Research center, 92916 Paris La Défense, France. \newline Acknowledgments: We thank the participants of the 2022 Mathematical and statistical methods for actuarial sciences and finance conference, Salerno, 2022 Workshop on empirical modelling of financial market participants, Paris-Saclay, 2022 MaxEnt conference, Paris, and 2022 online Econophysics colloquium, for useful comments.}} 

\date{\today}

\begin{document}

\maketitle

\begin{abstract}
We determine the amount of information contained in a time series of price returns at a given time scale, by using a widespread tool of the information theory, namely the Shannon entropy, applied to a symbolic representation of this time series. By deriving the exact and the asymptotic distribution of this market information indicator in the case where the efficient market hypothesis holds, we develop a statistical test of market efficiency. We apply it to a real dataset of stock indices, single stock, and cryptocurrency, for which we are able to determine at each date whether the efficient market hypothesis is to be rejected, with respect to a given confidence level.
\end{abstract}




\section{Introduction}

The efficient market hypothesis (EMH) is a cornerstone of financial theory. According to Eugene Fama, ``\textit{a market in which prices always fully reflect available information is called efficient}''~\cite{Fama}. One can in fact distinguish different natures of market efficiency, depending on how one defines the information set. The weak form of market efficiency considers that current prices reflect all the information contained in past prices. The semi-strong form includes, in addition, public information such as announcements of annual earnings or stock splits. The strong form also includes private information for some investors. 

In this work, we are interested in the weak-form efficiency. Our purpose is to determine whether this efficiency is a realistic assumption and therefore whether arbitrages may exist. In other words, we want to answer the question: can one take advantage of past prices to predict future evolutions of prices? However, one must distinguish two kinds of arbitrages. Pure arbitrage leads to a certain gain. Statistical arbitrage leads to an uncertain output which, in average only, is a gain. Since pure arbitrages are in practice almost instantaneously eliminated by the market, we are more interested here in statistical arbitrages. In this perspective, we define the market efficiency as the absence of statistical arbitrages with predictions based on past prices, that is as the characterization of prices by martingales. Therefore, determining whether the market is efficient or not, according to this definition, is of major importance both for quantitative portfolio managers, to help them define a universe of predictable assets to trade, and for market makers, to help them adjust their prices and avoid informed trades.

The financial literature puts forward several solutions to answer this question of the relevance of the EMH. The Hurst exponent is a widespread statistic for this purpose. As a parameter of a specific model, namely the fractional Brownian motion (fBm), it is related to the covariance between various price increments~\cite{MvN}. Therefore, given a particular value of the Hurst exponent, building predictions with this model is possible~\cite{NP,Garcin2017} and statistical arbitrages naturally follow~\cite{GMR,GNR,GarcinForecast}.

However, the fBm is not always the most realistic specification for depicting time series of log-prices. One may thus be inclined to use some extensions of the fBm, for example getting rid of the Gaussian distribution~\cite{ST,WBMW,GarcinMPRE,AG}, using time-varying parameters in a deterministic~\cite{BJR,PLV,ALV,Garcin2017,BP} or stochastic way~\cite{AT,BPP,GarcinMPRE}, or transforming the fBm in a stationary process~\cite{FBA,GarcinLamperti,GarcinEstimLamp}. All these extensions have fractal or multifractal properties, quantified by a Hurst exponent, whose interpretation is thus different from the fBm case and may not be related to the autocorrelation of the process. In this case, the Hurst exponent is not a relevant indicator of market efficiency and one therefore needs a statistic which is less related to a specific model. 

Information theory proposes some model-free alternatives to the Hurst exponent. In particular, a rich literature in econophysics applies concepts from this field to the quantification of the complexity in financial time series~\cite{SSS,BSS}. The most widespread measure of complexity is Shannon entropy which directly derives from a probability distribution, which can for example be the distribution of prices~\cite{SLOS,LBA} or of singular values of the matrix of lagged subsequences of prices~\cite{EPR}. One can also cite the approximate entropy which measures the complexity of dynamical systems~\cite{Pincus}, with applications in finance~\cite{PK,KV}. But one cannot easily relate the complexity of all these distributions to the notion of market efficiency. In this perspective, the complexity of two particular probability distributions is of greater interest. First, following the method proposed by Bandt and Pompe~\cite{BPompe}, the permutation entropy focuses on the distribution of the ordinal patterns of a time series of prices~\cite{ZTS,ZBB,BFF}. Second, Risso's method studies the discrete distribution of the sign of successive price returns~\cite{Risso,MBM,Ducournau}. The two approaches are close to each other in that they both consider that information resides in the order in which price increases and decreases occur.


In this literature on the application of information theory to measuring market efficiency, it is considered that the entropy is a gradual indicator of efficiency. The underlying idea is that the market is more or less efficient. We consider instead that the question about the relevance of the EMH is summarized into the following: can one predict future evolutions of prices with a ratio of winning bets significantly higher than $50\%$? Therefore, the statistical significance of entropy-based market efficiency indicators is crucial for their interpretation, beyond their gradual aspect. 

Focusing on Risso's method, the purpose of this work is to provide a statistical test of market efficiency. We thus compare the entropy of the market with the entropy of an ideal efficient market, the difference of entropies corresponding to the market information. We show that the estimator of this information is subject to a statistical error because of a limited number of observations. We however give the moments of the distribution of this estimator as well as a more synthetic and practical formula for its asymptotic distribution, which follows a gamma law. We study on simulations the extent to which this asymptotic distribution applies and we build a statistical test of market efficiency. A short empirical application illustrates that we are able to reject the EMH for some time series at certain times. This is a stronger conclusion than the simple comparison of efficiency levels usually put forward by the traditional literature on these entropy-related tools.

The rest of the paper is organized as follows. Section~\ref{sec:marketinfo} presents the market information, which stems from Risso's method. Section~\ref{sec:stat} presents the statistical perspectives about this market information indicator, including the estimator, its asymptotic distribution, and the statistical test for the EMH. Section~\ref{sec:empirical} shows a short empirical study to various financial series. Section~\ref{sec:conclusion} concludes.

\section{Market information}\label{sec:marketinfo}

We consider, for a given financial asset, a time series of $n+1$ consecutive prices $P_0,...,P_n$. In the empirical study, Section~\ref{sec:empirical}, we work with daily prices, but we can also consider other frequencies. We transform the series of prices in a binary series $X_1,...,X_n$, which indicates whether the price has increased or decreased between two consecutive observations:
$$X_i=\indic_{\{P_i-P_{i-1}>0\}}.$$
This symbolic representation of a series of consecutive prices is consistent with Risso's approach~\cite{Risso}. For example the sequence $(X_1,X_2,X_3)=(0,1,1)$ represents a decrease followed by two daily price increases. For a given length $L<n$, $2^L$ sequences are possible. We order them, for example, with Gray's binary code. If $L=3$, the 8 possible sequences are 
$$(G^3_1,...,G^3_8)=((0,0,0), (0,0,1), (0,1,1), (0,1,0), (1,1,0), (1,0,0), (1,0,1), (1,1,1)).$$

The purpose of Shannon entropy is to determine the amount of uncertainty in a discrete probability distribution. In Risso's approach, one focuses on the distribution of the sequences of increase indicators, given a length $L$. We assume that the random series $(X_i)$ is stationary. The probability to draw a particular sequence of length $L$ is noted $p^L_i=\proba((X_{.},X_{.+1},...,X_{.+L-1})=G^L_i)$. The Shannon entropy of the discrete distribution $(p^L_i)_i$ is defined by:
$$H^L=-\sum_{i=1}^{2^L} p^L_i\log_2\left(p^L_i\right),$$
where we use the convention that $0\log_2(0)=0$ which is reasonable because $\lim\limits_{x\to 0} x\log_2(x)=0$. We note that the highest entropy corresponds to a uniform distribution: $\forall i\in\llbracket 1,2^L\rrbracket$, $p^L_i=2^{-L}$. 

In Risso's approach, a normalized version of $H^L$, thus belonging to $[0,1]$, plays the role of efficiency indicator. The larger this indicator, the more efficient the market. This leads to a gradual interpretation of market efficiency. On the contrary, the approach we put forward is based on a binary interpretation of efficiency: the market is efficient or it isn't. We therefore propose an indicator which is slightly different from that of Risso. First, we don't need to normalize Shannon entropy. Second, instead of considering the entropy, we use the amount of information contained in the distribution of sequences of increase indicators $X_i$, with respect to what would be this distribution according to the EMH. We thus need to clearly define what are the distributions consistent with the EMH. In this perspective, we find more natural to base our analysis on conditional distributions instead of on the non-conditional distributions used in Risso's approach. Indeed, in our symbolic framework, the statistical arbitrages appearing in an inefficient market take advantage of a difference of probability between the two possible values of the future increase indicator $X_{L+1}$, conditionally on observed past indicators $(X_1,...,X_{L})$.

Following this conditional approach, we decompose a sequence $(X_1,...,X_{L+1})$ observed in time $L$ in two parts. The $L$ first elements constitute a prefix sequence, corresponding to the last $L$ observed increase indicators at the current time. The suffix, $X_{L+1}$, corresponds to the next, unobserved and random increase indicator. The discrete probability of the prefix $(X_1,...,X_{L})=G^{L}_i$ is $p^{L}_i$, as already mentioned. Conditionally on this prefix, the distribution of the suffix $X_{L+1}$ is Bernoulli of a parameter noted $\pi^{L}_i\in[0,1]$. Of course, $\pi^{L}_i$ is only defined if $p^{L}_i\neq 0$. The full sequence $(X_1,...,X_{L+1})$ is thus equal to $(G^{L}_i,1)$ with probability $p^{L}_i\pi^{L}_i$ and $(G^{L}_i,0)$ with probability $p^{L}_i(1-\pi^{L}_i)$. We then get the Shannon entropy of the discrete distribution of this full sequence of length $L+1$:
\begin{equation}\label{eq:entropyPrefSuff}
H^{L+1}=-\sum_{i=1}^{2^{L}} \left(p^{L}_i\pi^{L}_i\log_2\left(p^{L}_i\pi^{L}_i\right) + p^{L}_i(1-\pi^{L}_i)\log_2\left(p^{L}_i(1-\pi^{L}_i)\right)\right).
\end{equation}

The EMH asserts that, conditionally on the prefix $(X_1,...,X_{L})$, the two possible values for the suffix, $X_{L+1}=1$ and $X_{L+1}=0$, have equal probability. The Shannon entropy $H^{L+1}_{\star}$ for a market following the EMH thus expresses as a particular case of equation~\eqref{eq:entropyPrefSuff}, with $\pi^{L}_i=1/2$: 
\begin{equation}\label{eq:Hstar}
H^{L+1}_{\star}=-\sum_{i=1}^{2^{L}} p^{L}_i\log_2\left(\frac{p^{L}_i}{2}\right)=1+H^{L}.
\end{equation}
Finally, we define our efficiency indicator, the market information, as the difference between the entropy consistent with the ideal EMH and the true entropy of the market:
\begin{equation}\label{eq:marketinfo}
I^{L+1}=H^{L+1}_{\star}-H^{L+1}.
\end{equation}
In Section~\ref{sec:stat}, we use the estimated value of this indicator to build a statistical test of market efficiency. This is possible because of the important following theorem, which shows that the value of $I^{L+1}$ discriminates efficient and inefficient markets.

\begin{thm}\label{th:positive}
For $L\in\mathbb N^{\star}$, we have $I^{L+1}\geq 0$. Moreover,
$$I^{L+1}=0 \Longleftrightarrow \forall i\in\llbracket 1,2^L\rrbracket\setminus\left\{j| p^L_j=0\right\}, \pi^{L}_i=\frac{1}{2}.$$
\end{thm}

The proof is postponed in Appendix~\ref{sec:proofpositive}.

Theorem~\ref{th:positive} states that if the market follows the EMH, then $I^{L+1}$ is equal to 0. If the market does not follow the EMH by leading to some $\pi^{L}_i\neq 1/2$, then $I^{L+1}>0$. Nevertheless, we are aware that some time series of prices may have all their $\pi^{L}_i$ equal to $1/2$, while they are not consistent with the EMH. The reason for this is that we have summarized a price increment in $X_i$, which can take only two values, 0 or 1. Some information contained in the time series is lost with this two-state symbolic approach. One could thus be tempted to generalize Risso's approach to a greater number of states or even to a continuum of states~\cite{SLOS,LBA}. We leave this possible extension to further research and focus on the classical two-state approach, which we believe to be enlightening. We however underline below some challenges regarding the extension to more than two states. 

First, we recall that the aim of our work is to build a statistical test of market efficiency, in which we determine whether we can reject the null hypothesis corresponding to the EMH. If we reject the EMH in our simplified two-state approach with a given confidence (typically $99\%$, e.g.), we know that we are able to reject the EMH in a more realistic framework with an even greater confidence. In other words, using only two states tends to diminish the power of the statistical test but should not alter its significance level. Second, if we include more states than only two, it is more difficult to link entropy to the notion of market efficiency. Indeed, with only two states, it is clear that the efficient market corresponds to the uniform conditional distribution of the suffix ($\pi^L_i=1/2$). With a greater number of states, the definition of the set of distributions consistent with the EMH is not as trivial and Shannon entropy may not be an appropriate quantity for discriminating efficient and inefficient markets. In particular, outside the two-state framework, a same level of entropy may correspond to several distributions, among which some but not all are consistent with the EMH. Third, given a prefix sequence, increasing the number of states will decrease the number of observations of each suffix and thus increase the variance of the estimator of the probability of each state. In other words, the statistical significance of the estimated market information may be reduced.

\section{Statistical perspective}\label{sec:stat}

We now focus on the statistical properties of the concept of market information defined in Section~\ref{sec:marketinfo}. We begin by introducing a simple estimator along with some of its properties. Then we determine the asymptotic distribution of this estimator and we use these elements to build a statistical test of market efficiency.

\subsection{Estimator of the market information}\label{sec:estimator}

The definition of the market information, as displayed in equation~\eqref{eq:marketinfo}, relies on unobserved probabilities $p^L_i$ and $\pi^L_i$. Replacing these theoretical probabilities by their empirical version leads to an estimator $\widehat{I}^{L+1}$ of the market information $I^{L+1}$. We note that it is the value of $\pi^L_i$, not the value of $p^L_i$, which makes it possible to conclude about the relevance of the EMH. We will thus be particularly interested in the behaviour of the market information with respect to the estimated conditional probabilities $\pi^L_i$. Consequently, most of the results below will be stated conditionally on the probabilities $p^L_i$.

A first remark regarding $\widehat{I}^{L+1}$ is that this estimate also follows Theorem~\ref{th:positive} in which probabilities are replaced by empirical probabilities. In other words, we still have $\widehat{I}^{L+1}\geq 0$. But if the EMH holds, we may have $\widehat{I}^{L+1}$ \textit{slightly} higher than 0 because, due to the statistical error, the estimated probabilities $\widehat{\pi}^L_i$ may be different from $\pi^L_i$, which are $1/2$ in this case. Building a statistical test is thus mandatory for answering the question of the efficiency of the market. It will be the purpose of Section~\ref{sec:test}.

Considering that the theoretical value of the market information is 0, according to the EMH, and that the estimator $\widehat{I}^{L+1}$ is nonnegative and may be different from zero even when the EMH holds, we conclude that $\widehat{I}^{L+1}$ is biased. We will see however in Section~\ref{sec:asympt} that the bias tends to zero when $n$ tends to infinity, at least in the EMH case.

We provide some insight on the distribution of the estimator $\widehat{I}^{L+1}$ with its exact moment-generating function, displayed in the following proposition. We focus on the case where the EMH holds, that is for $\pi^L_i=1/2$ for all $i$, but for an empirical market information based on $\widehat{\pi}^L_i$ instead of $\pi^L_i$. The provided expression also assumes that we have access to the true probabilities of the prefixes: $\widehat{p}^L_i=p^L_i$.

\begin{pro}\label{pro:mgf_exact}
For $L\in\mathbb N^{\star}$, the moment-generating function of $\widehat{I}^{L+1}$, conditionally on the event $\mathcal E= \{\forall i\in\llbracket 1,2^L\rrbracket, \widehat{p}^L_i=p^L_i,\pi^L_i=1/2\}$, is:
$$M_{\widehat{I}^{L+1}}:t\mapsto\E\left[\left.e^{t\widehat{I}^{L+1}}\right|\mathcal E\right]=e^t\prod_{i=1}^{2^L} \sum_{j=0}^{n_i} C^L_{i,j}(t),$$
for values of $t$ for which this quantity is defined and where 
$$C^L_{i,j}(t)=\binom{n_i}{j}\frac{1}{2^{n_i}} \left(\frac{j}{n_i}\right)^{t p^L_i j/n_i \ln(2)} \left(1-\frac{j}{n_i}\right)^{t p^L_i (1-j/n_i)/ \ln(2)}.$$
\end{pro}

The proof of this proposition is postponed in Appendix~\ref{sec:proofmgf_exact}.

We can then deduce from Proposition~\ref{pro:mgf_exact} the moments of the estimated market information, like in the following proposition, whose proof is in Appendix~\ref{sec:proofmoment}.

\begin{pro}\label{pro:moment}
For $L\in\mathbb N^{\star}$ and $r\in\mathbb N$, the conditional moment of order $r$ of the estimator of market information is
$$\E\left[\left.\left(\widehat{I}^{L+1}\right)^r\right|\mathcal E\right] = \sum_{m=0}^r \binom{r}{m} \sum_{j_1=0}^{n_1} \ldots\sum_{j_{2^L}=0}^{n_{2^L}}\alpha_{j_1,...,j_{2^L}}^{m}\prod_{i=1}^{2^L} \binom{n_i}{j_i}\frac{1}{2^{n_i}} ,$$
where 
$$\alpha_{j_1,...,j_{2^L}}=\sum_{k=1}^{2^L} p^L_k\left(\frac{j_k}{n_k}\log_2\left(\frac{j_k}{n_k}\right)+\left(1-\frac{j_k}{n_k}\right)\log_2\left(1-\frac{j_k}{n_k}\right)\right).$$
\end{pro}

The formula with the nested sum in Proposition~\ref{pro:moment} leads to a slow computation of moments, in particular when $L$ is big. We can however also provide a simpler expression of the moments of low order, for example
$$\begin{array}{ccl}
\E\left[\left.\widehat{I}^{L+1}\right|\mathcal E\right] & = & \sum_{i=1}^{2^L} \left[-p^L_i\log_2\left(\frac{p^L_i}{2}\right) \right. \\
 & & \left.+p^L_i\sum_{j_i=0}^{n_i} \binom{n_i}{j_i}\frac{1}{2^{n_i}} \left(\frac{j_i}{n_i}\log_2\left(p^L_i\frac{j_i}{n_i}\right)+\left(1-\frac{j_i}{n_i}\right)\log_2\left(p^L_i\left(1-\frac{j_i}{n_i}\right)\right)\right)\right] \\
 & = & 1+\sum_{i=1}^{2^L} p^L_i2^{1-n_i} \sum_{j_i=0}^{n_i} \binom{n_i}{j_i} \frac{j_i}{n_i}\log_2\left(\frac{j_i}{n_i}\right).
\end{array}$$
This expression relies on the direct calculation of the first moment, without using the moment-generating function. Nevertheless, the moment-generating function or any similar transform is useful to fully describe the probability distribution of the estimator of market efficiency. A more condensed asymptotic expression of such a transform would thus be helpful in the perspective of a practical application. It is the purpose of the next section.

\subsection{Asymptotic analysis}

We decompose the study of the asymptotic distribution of the empirical market information in two steps: first we focus on the summands, then on the market information itself.

\subsubsection{Summands of the market information}

We introduce the function $g_j$, defined, for $(t,x)\in\mathbb R\times(0,1)$ and a given $j\in\llbracket 1,2^L\rrbracket$, by:
\begin{equation}\label{eq:gj}
g_j(t,x)=\exp\left(i t\left[p^L_j x\log_2(p^L_j x)+p^L_j (1-x)\log_2(p^L_j (1-x))-p^L_j\log_2\left(\frac{p^L_j}{2}\right)\right]\right),
\end{equation}
where $i$ is the imaginary unit. As related to a summand of $\widehat I^{L+1}$, this functions appears in the characteristic function of the market information:
\begin{equation}\label{eq:caracgj}
\varphi_{\widehat I^{L+1}}:t\in\mathbb R\mapsto\E\left[\left. e^{it\widehat I^{L+1}}\right|\mathcal E\right] = \prod_{j=1}^{2^L} \E\left[g_j(t,\widehat \pi_j)\right].
\end{equation}
Before giving an asymptotic expression for the characteristic function, we provide some useful properties on the function $g_j$, beginning by an expression of its derivatives.

\begin{pro}\label{pro:deriv_gj}
For $k\geq 1$ and $j\in\llbracket 1,2^L\rrbracket$, the $k$-th derivative with respect to $x\in(0,1)$ of the function $g_j$ defined in equation~\eqref{eq:gj}, is
\begin{equation}\label{eq:deriv_gj}
\partial_x^k g_j(t,x)=g_j(t,x)\sum_{l=1}^{k}\left(\frac{itp^L_j}{\ln(2)}\right)^l B_{k,l}\left(\lambda(x),\frac{d}{dx}\lambda(x),...,\frac{d^{k-l}}{dx^{k-l}}\lambda(x)\right),
\end{equation}
where $B_{k,l}$ is a Bell polynomial and $\lambda(x)=\ln(x)-\ln(1-x)$. In particular, when $x=1/2$, we have
\begin{equation}\label{eq:deriv_gj_1/2}
\partial_x^k g_j\left(t,\frac{1}{2}\right)=2^k\sum_{l=1}^{k}\left(\frac{itp^L_j}{\ln(2)}\right)^l B_{k,l}\left(0,0!,0,2!,0,4!,0,...\right)
\end{equation}
and, in the particular case where $k$ is odd, $\partial_x^k g_j(t,1/2)=0$.
\end{pro}

The proof of this proposition is postponed in Appendix~\ref{sec:proofderiv_gj}.

It is also easy to see, from equation~\eqref{eq:gj}, that for all $y\in[0,1/2)$, $g_j(t,y+1/2)=g_j(t,-y+1/2)$. As a consequence, 
\begin{equation}\label{eq:sym_gj}
\partial_x^k g_j\left(t,y+\frac{1}{2}\right)=(-1)^k\partial_x^k g_j\left(t,-y+\frac{1}{2}\right).
\end{equation}
So the value of the function $x\mapsto |\partial_x^k g_j(t,x)|$ evolves symmetrically with respect to $x=1/2$. The next proposition gives some insight on the amplitude, understood as the $L^q$ norm of the derivatives of $g_j$ restricted to an interval whose left bound is $1/2$. It can be easily extended to integrals on an interval centred in $1/2$ thanks to equation~\eqref{eq:sym_gj}.

\begin{pro}\label{pro:average_deriv_gj}
For $k\geq 1$, $j\in\llbracket 1,2^L\rrbracket$, $q\geq 1$, $z\in[1/2,1)$, and $t\in\mathbb R$, the $L^q([1/2,z])$ norm of the function $g_j$, defined in equation~\eqref{eq:gj}, admits the following bound:
$$\left(\int_{1/2}^{z}{\left|\partial_x^k g_j(t,x)\right|^q dx}\right)^{1/q} \leq \sum_{l=1}^{k}\left|\frac{tp^L_j}{\ln(2)}\right|^{l} \mathcal L(k,l) \left(\frac{2}{(r(k-l)+1)r(k-l)}\right)^{l} \frac{(1-z)^{-lr(k-l)+1/q}}{(qlr(k-l)-1)^{1/q}},$$
where $\mathcal L(k,l)$ is a Lah number, equal to ${k-1 \choose l-1} \frac{k!}{l!}$, and $r:\mathbb N\rightarrow\mathbb N$ is defined by:
\begin{equation}\label{eq:rk}
r(k)=\max\left(5,2\left\lceil\frac{k-1}{2}\right\rceil+1\right).
\end{equation}
\end{pro}

The proof of this proposition is postponed in Appendix~\ref{sec:proofaverage_deriv_gj}.

The following theorem provides an approximation for $\E[g_j(t, X/n_j)]$, where $X$ is a binomial variable of parameters $n_j$ and $1/2$, with an upper bound for the error. This approximation will then be useful for deriving the characteristic function of the market information, as defined in equation~\eqref{eq:caracgj}.

\begin{thm}\label{th:Taylor_gj}
For $j\in\llbracket 1,2^L\rrbracket$, $t\in\mathbb R$, $g_j$ defined in equation~\eqref{eq:gj}, and $X\sim\mathcal B(n_j,1/2)$, we have
$$\E\left[g_j\left(t,\frac{X}{n_j}\right)\right]=1+\frac{itp^L_j}{2\ln(2)n_j}+R(t,n_j),$$
where, for all $q\in\mathbb N\setminus\{0,1\}$, we have for all $\varepsilon>1$, the existence of $\nu\in\mathbb N$ such that, for all $n_j\geq\nu$,
\begin{equation}\label{eq:borneerreurgj}
|R(t,n_j)|\leq\frac{\varepsilon}{96} \left(q-1\right)^{1-1/q} (4q-1)^{3}
\left(\sum_{l=1}^{4}\left|\frac{2^5tp^L_j}{15\ln(2)}\right|^{l} \frac{\mathcal L(4,l)}{(5ql-1)^{1/q}}\right) n_j^{-2+1/2q}.
\end{equation}
\end{thm}

The proof of this theorem is postponed in Appendix~\ref{sec:Taylor_gj}.

The upper bound of the error term in Theorem~\ref{th:Taylor_gj} depends on two free parameters, $\varepsilon$ and $q$. It is obvious that the lower possible value for $\varepsilon$, the tighter the bound. When $n_j\rightarrow\infty$, replacing $\varepsilon$ by 1 in equation~\eqref{eq:borneerreurgj} provides an asymptotic expression of the upper bound of $|R(t,n_j)|$. Regarding the selection of $q$, it is directly related to the convergence speed with respect to $n_j$ because of the term $n_j^{-2+1/2q}$. On the one hand, a higher $q$ increases the convergence rate of the error term with respect to $n_j$, on the other hand it also increases the constant term in the formula of the upper bound. Figure~\ref{fig:ErrorBound} displays a numerical evaluation of the asymptotic bound of Theorem~\ref{th:Taylor_gj} for several possible values for $q$. It finally suggests that the lower $q$, the lower the bound, regardless of $n_j$. With $q=2$, the error term is asymptotically $\mathcal O(n_j^{-1.75})$ which is to be compared to the $n_j^{-1}$ appearing in the approximation of $\E[g_j(t, X/n_j)]$ in Theorem~\ref{th:Taylor_gj}.

\begin{figure}[htbp]
  \centering
  \includegraphics[width=0.7\linewidth]{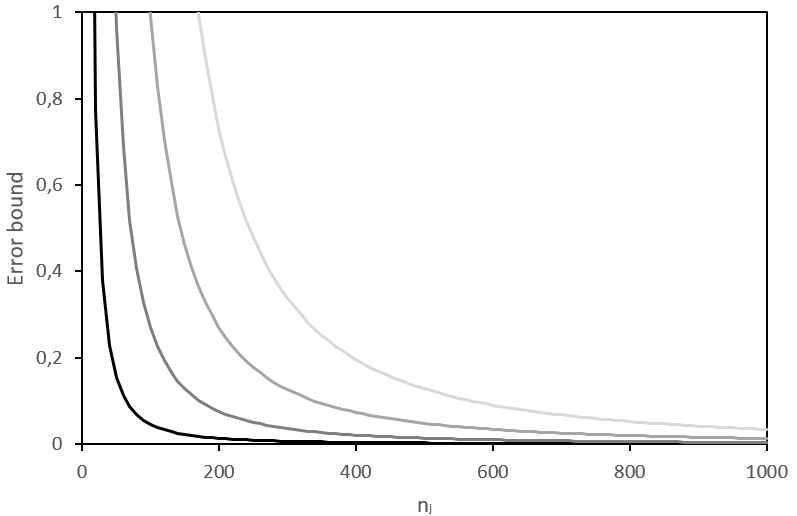}
  \begin{minipage}{0.7\textwidth}\caption{Assymptotic upper bound of $|R(t,n_j)|$ according to Theorem~\ref{th:Taylor_gj}, for several values of $q$, namely, from the bottom to the top: 2, 3, 4, 5. We have considered $p^L_j=1/2$ and $t=1$.}\label{fig:ErrorBound}
  \end{minipage}
\end{figure}

\subsubsection{Asymptotic distribution of the market information estimator}\label{sec:asympt}

We are now interested in the characteristic function of the empirical market information, conditionally on $\mathcal E$. According to equation~\eqref{eq:caracgj}, and using Theorem~\ref{th:Taylor_gj}, we have, for $t\in\mathbb R$ and $p^L_j=n_j/n>0$ whatever $j$,
\begin{equation}\label{eq:caracmarketinfo}
\begin{array}{ccl}
\varphi_{\widehat I^{L+1}}(t) & = & \prod_{j=1}^{2^L} \E\left[g_j(t,\widehat \pi_j)\right] \\
 & \overset{n\rightarrow\infty}{\sim} & \prod_{j=1}^{2^L} \left(1+\frac{it}{2\ln(2)n}\right) \\
  & \overset{n\rightarrow\infty}{\sim} & \left(1-\frac{it}{\ln(2)n}\right)^{-2^{L-1}}.
\end{array}
\end{equation}
We recognize the characteristic function of the gamma distribution $\Gamma(k,\theta)$ of shape parameter $k=2^{L-1}$ and scale parameter $\theta=1/\ln\left(2\right)n$. We note that the result does not depend on the specific value of each $p^L_j$, but only on $n$ and $L$.

We underline that a related literature on transfer entropy concludes that the information, defined as a difference of conditional entropies, follows asymptotically a chi-square distribution, using Wilks' theorem~\cite{BB,KS,BC}. While the framework and definition of information are slightly different from ours, we note that we would also obtain such an asymptotic chi-square distribution if we changed the scale of our market information, in particular by replacing the base 2 of the logarithm by a base $e$ in equations~\eqref{eq:entropyPrefSuff} and~\eqref{eq:Hstar}. Morover, beyond the asymptotic distribution, we have also contributed to give some insight into the way the error of the approximation behaves, thanks to Theorem~\ref{th:Taylor_gj}.

Figures~\ref{fig:distrib_100} and~\ref{fig:distrib_I3} confirm with the help of simulations the relevance of the gamma distribution, with the parameters $k=2^{L-1}$ and $\theta=1/\ln\left(2\right)n$, as the asymptotic distribution of the empirical market information under the hypothesis of an efficient market, for two values of $L$. However, this confirmation is only visual. In order to study more thoroughly the accuracy of our gamma approximation, we conduct a statistical test to assess whether the very slight difference one can see between the simulated and asymptotic cumulative distribution functions in Figures~\ref{fig:distrib_100} and~\ref{fig:distrib_I3} is significant or not, depending on the value of $n$. We answer this question with a Kolmogorov-Smirnov test, whose conclusion, for $L=1$, is that the asymptotic distribution is valid for $n\gtrsim 100$, as displayed in Figure~\ref{fig:KS_stat}.

\begin{figure}[htbp]
	\centering
		\includegraphics[width=0.45\textwidth]{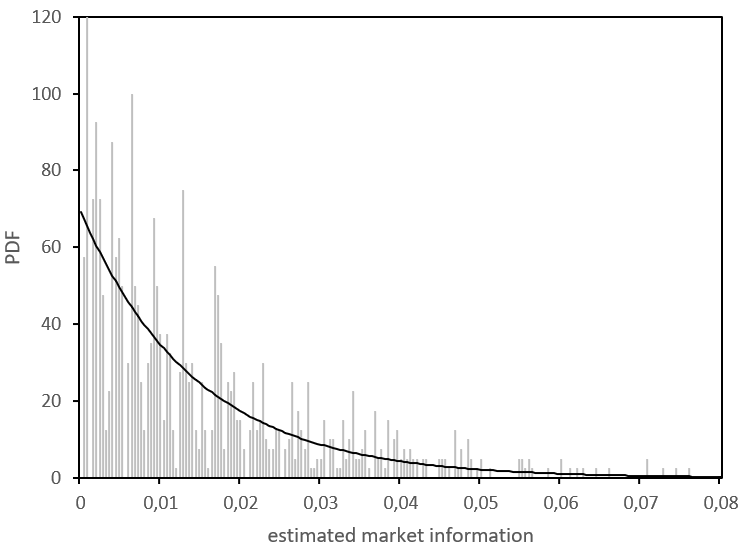} 
		\includegraphics[width=0.45\textwidth]{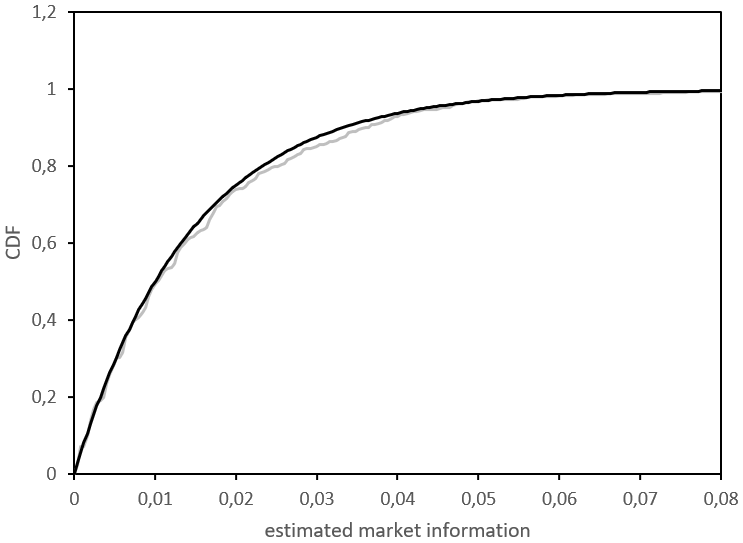} 
\begin{minipage}{0.9\textwidth}\caption{Asymptotic $\Gamma(2^{L-1},1/\ln\left(2\right)n)$ (black) and simulated (on 1,000 trajectories, grey) distributions of the estimated market information $\widehat{I}^2$ for $n=100$.}
	\label{fig:distrib_100}
\end{minipage}
\end{figure}

\begin{figure}[htbp]
	\centering
		\includegraphics[width=0.45\textwidth]{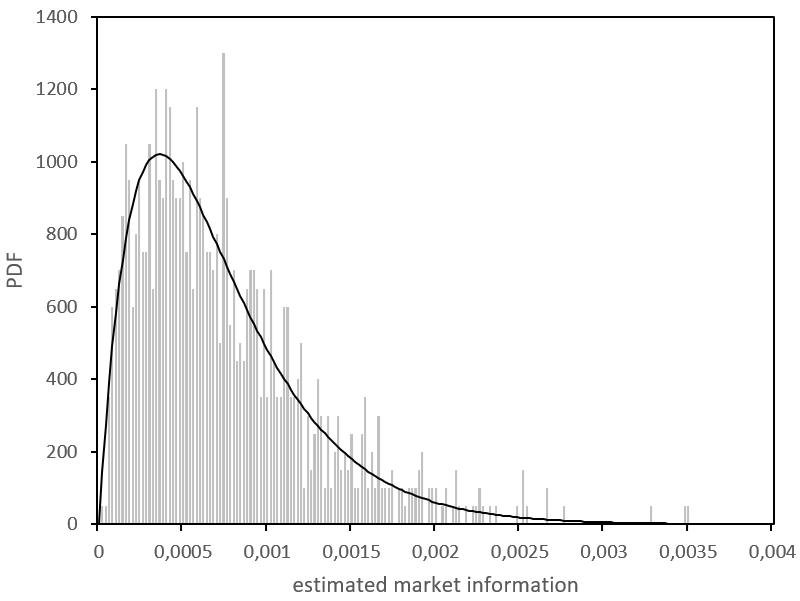} 
		\includegraphics[width=0.45\textwidth]{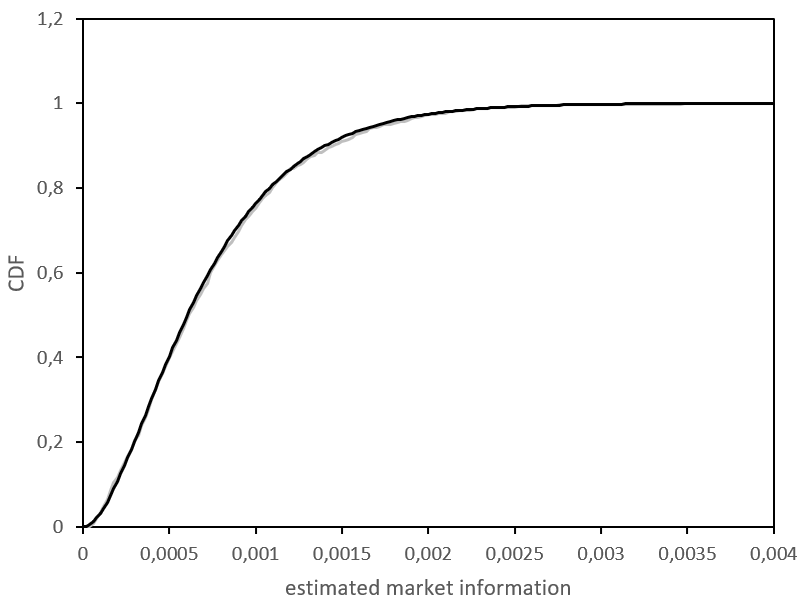} 
\begin{minipage}{0.9\textwidth}\caption{Asymptotic $\Gamma(2^{L-1},1/\ln\left(2\right)n)$ (black) and simulated (on 1,000 trajectories, grey) distributions of the estimated market information $\widehat{I}^3$ for $n=4,000$.}
	\label{fig:distrib_I3}
\end{minipage}
\end{figure}

\begin{figure}[htbp]
	\centering
		\includegraphics[width=0.65\textwidth]{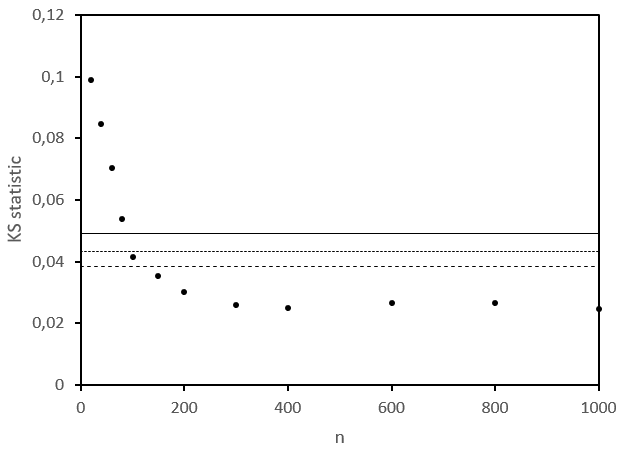} 
\begin{minipage}{0.9\textwidth}\caption{Kolmogorov-Smirnov statistic between the asymptotic and the simulated (on 1,000 trajectories) distributions of the estimated market information $\widehat{I}^2$ for various values of $n$. The three lines indicate the values of the statistics leading to a p-value of $5\%$, $1\%$, and $0.1\%$ (from bottom to top).}
	\label{fig:KS_stat}
\end{minipage}
\end{figure}

\subsection{Statistical test of market efficiency}\label{sec:test}

The purpose of our work is to statistically test the existence of market efficiency using the empirical market information as a statistic. More precisely, given $L\in\mathbb N$, we define the null hypothesis by $\forall i\in\llbracket 1,2^L\rrbracket\setminus\left\{j| p^L_j=0\right\}, \pi^{L}_i=1/2$. The alternative hypothesis thus corresponds to the existence of an $i$ such that $\pi^{L}_i\neq 1/2$. One could indeed use this imbalance to do statistical arbitrages. As already exposed, the true definition of market efficiency is more restrictive than the formalism we chose in our null hypothesis. Nevertheless, we think that this approach is interesting insofar as rejecting our null hypothesis for a given confidence leads to rejecting the market efficiency with an even higher confidence.

Thanks to Theorem~\ref{th:positive}, the null hypothesis is equivalent to having the market information $I^{L+1}$ equal to zero. We evaluate this quantity with the statistic $\widehat{I}^{L+1}$. Unfortunately, under the null hypothesis, there is a bias in $\widehat{I}^{L+1}$, which asymptotically disappears according to equation~\eqref{eq:caracmarketinfo}. We have also provided in Section~\ref{sec:asympt} an asymptotic gamma distribution for $\widehat{I}^{L+1}$ under the null hypothesis, which we can use to quantify the p-value corresponding to the obtained test statistic and determine whether we are able or not to reject the EMH. Figure~\ref{fig:Info_n} reports the link between the value of the test statistic, its p-value, and $n$. We note that an alternative to our approach has previously been published, in which the confidence intervals were estimated by simulations instead of our straightforward asymptotic distribution~\cite{SMM}.

\begin{figure}[htb]
	\centering
		\includegraphics[width=0.65\textwidth]{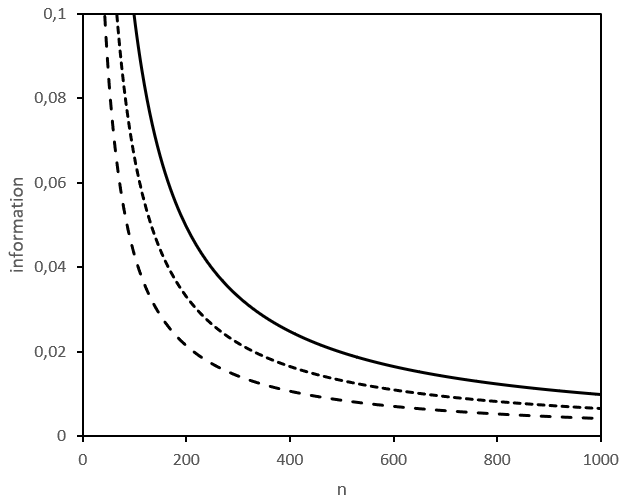} 
\begin{minipage}{0.9\textwidth}\caption{Information corresponding to a p-value of $5\%$, $1\%$, and $0.1\%$ (from bottom to top) in the test with a market information $I^2=0$ as a null hypothesis, for various values of $n$.}
	\label{fig:Info_n}
\end{minipage}
\end{figure}

\section{Empirical study}\label{sec:empirical}

We now apply the method introduced above to real financial data. We focus on one year of daily data, 2021, for four time series of prices: two stock indices, the CAC 40 index and Russell 2000 index, the stock of Perficient, which is a constituent of Russell 2000 index, and the cryptocurrency BTC/USD. 

Using the estimator $\widehat{I}^{L+1}$ introduced in Section~\ref{sec:estimator}, we evaluate the market information for each of these series, using the last $n$ observed daily price returns. We focus on a particular setting with $L=1$ and $n=100$ business days. We thus display the estimated market information between end May 2021 and end December 2021 in Figure~\ref{fig:Info_real}. Note that the BTC/USD has more trading days than the three other series, so that our series of market information in this case is longer and starts the 11th April 2021. 

\begin{figure}[htb]
	\centering
		\includegraphics[width=0.45\textwidth]{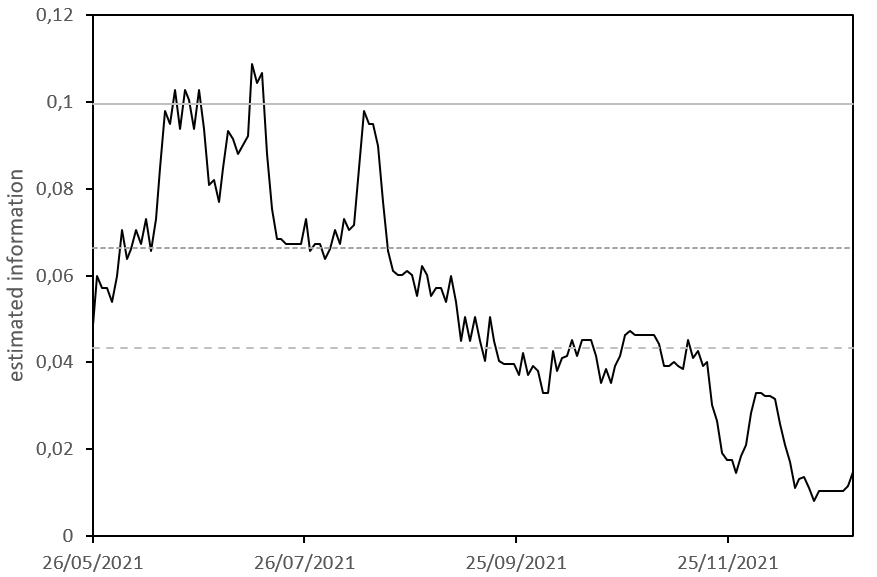}
		\includegraphics[width=0.45\textwidth]{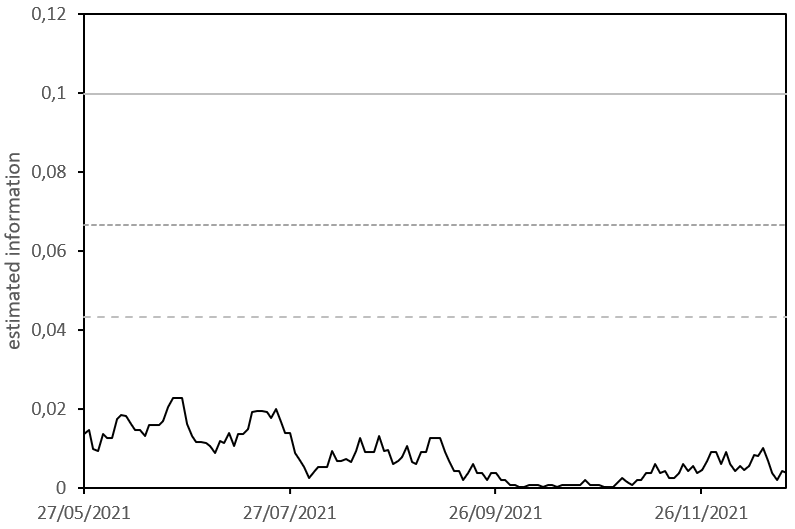} \\
		\includegraphics[width=0.45\textwidth]{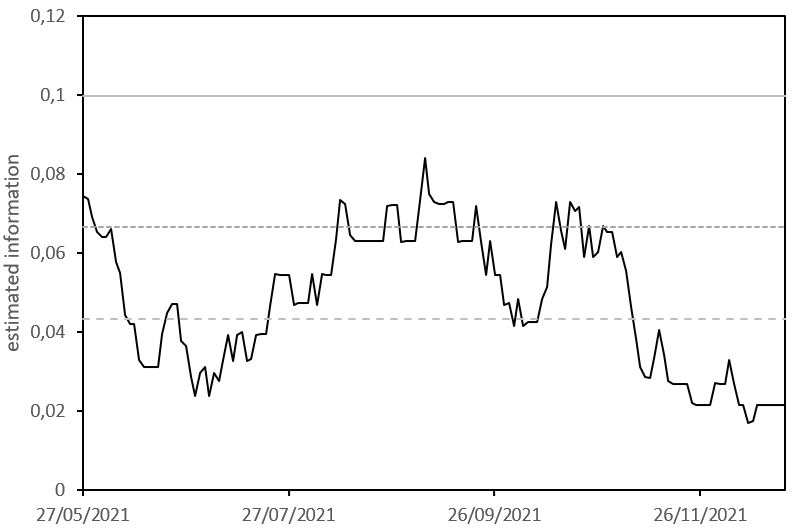} 
		\includegraphics[width=0.45\textwidth]{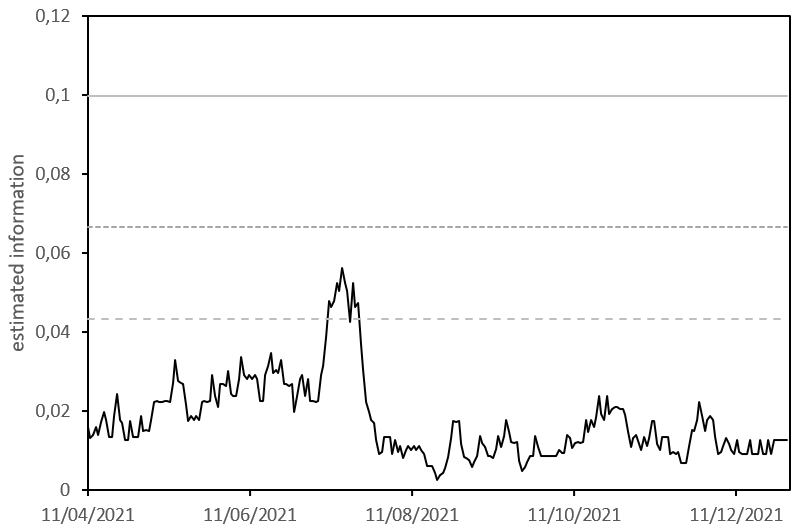} 
\begin{minipage}{0.9\textwidth}\caption{Estimated information $\widehat{I}^2$ in a rolling window of $n=100$ business days, for the CAC 40 index (top left), Russell 2000 index (top right), Perficient stock (bottom left), and BTC/USD (bottom right). The grey lines are the confidence intervals of the statistical test of absence of information, with probabilities of $95\%$, $99\%$, and $99.9\%$.}
	\label{fig:Info_real}
\end{minipage}
\end{figure}

We observe in Figure~\ref{fig:Info_real} some differences between the time series. For instance, the market information in our dataset is lower for Russell 2000 index and BTC/USD than for CAC 40 index and Perficient stock. But the market information is also time-varying and the rank of each of these four assets with respect to their information may change. It is for example the case for the CAC 40 index and Perficient stock: the French stock index had a higher information in June than Perficient stock while we observe the opposite late October.

Besides the gradual interpretation of the market information, we are interested in applying the statistical test of market efficiency exposed in Section~\ref{sec:test}. According to our results gathered in Figure~\ref{fig:Info_real}, we are never able to reject the null hypothesis of market efficiency for Russell 2000 index, whereas we reject it for the CAC 40 index with a confidence of $95\%$ (respectively $99\%$ and $99.9\%$) at $60.8\%$ (resp. $31.0\%$ and $4.4\%$) of the dates at which we have estimated the market information. Opportunities of statistical arbitrages thus seem more frequent for the CAC 40 index than for the Russell 2000 index. This difference between the two indices may be due to the way each index is built and particularly to their number of constituents, which is 50 times higher for Russell index. When focusing on a single constituent of Russell 2000 index, namely Perficient stock, we can reject the EMH $57.0\%$ (respectively $15.2\%$) of the time with a confidence of $95\%$ (resp. $99\%$). This shows that opportunities of statistical arbitrages may exist for the constituents of an index, whereas mixing many stocks in an index makes these opportunities disappear at the scale of the index.

Regarding the cryptocurrency market, which is known to follow very specific dynamics, most of the time we cannot reject the EMH. We can only reject it, with a confidence of $95\%$, $4.2\%$ of the time.

\section{Conclusion}\label{sec:conclusion}

Determining whether the EMH holds or not for a particular financial asset is of major importance in the asset management industry. It can indeed help market agents in their investment decisions. In this article, we have proposed a statistical test of market efficiency which may be used in this practical perspective. This statistical test is based on a market information estimator relying on Shannon entropy and a symbolic representation of a series of successive price returns. Applying this tool to real financial data shows a diversity among financial markets regarding their efficiency. Future research in this field could focus on extending the statistical test introduced above to other information statistics, such as the permutation entropy~\cite{BPompe}.

\bibliographystyle{plain}

\bibliography{biblioEff}

\appendix

\section{Proof of Theorem~\ref{th:positive}}\label{sec:proofpositive}

\begin{proof}
Let the function $f$ be defined by $f:(x,y)\in(0,1]\times(0,1)\mapsto xy\log_2(xy) + x(1-y)\log_2(x(1-y))-x\log_2(x/2)$. We extend $f$ by continuity for $y\in[0,1]$ noting that $f(x,0)=f(x,1)=x$. Obtaining some of its derivatives for $x\in(0,1]$ and $y\in(0,1)$ is straightforward:
\begin{equation}\label{eq:preuvepositive}
\left\{\begin{array}{ccl}
\frac{\partial f}{\partial x}(x,y) & = & y\log_2(xy) + (1-y)\log_2(x(1-y))-\log_2(x/2) \\
\frac{\partial^2 f}{\partial y\partial x}(x,y) & = &  \log_2\left(\frac{y}{1-y}\right).
\end{array}\right.
\end{equation}
From these expressions, we note that $f$, $\partial f/\partial x$, and $\partial^2 f/\partial y\partial x$ are all equal to zero when evaluated in $y=1/2$, regardless of the value of $x\in(0,1]$. Moreover, when $y>1/2$ (respectively $<1/2$), we have $y/(1-y)>1$ (resp. $<1$) and thus $\partial^2 f/\partial y\partial x$ is positive for $y\in(1/2,1)$ and negative for $y\in(0,1/2)$. Its primitive $y\mapsto\partial f/\partial x (x,y)$ is thus strictly decreasing in $(0,1/2)$ and strictly increasing in $(1/2,1)$, with a value of 0 in $1/2$. Consequently, $\partial f/\partial x$ is nonnegative for $y\in(0,1)$, and $1/2$ is the only value of $y$ for which it is equal to zero. From equation~\eqref{eq:preuvepositive}, we observe that, for all $(x,y)\in(0,1]\times(0,1)$,
$$f(x,y)=x\frac{\partial f}{\partial x}(x,y).$$
Using in addition the extension of $f$ to $y\in[0,1]$, we can thus conclude that $f$ is nonnegative for $y\in[0,1]$, and $1/2$ is the only value of $y$ for which it is equal to zero, regardless of $x\in(0,1]$. Going back to the market information formula, we note that
$$I^{L+1}=\sum_{i=1}^{2^{L}} f\left(p^{L}_i,\pi^{L}_i\right).$$
Each element of this sum is positive, as soon as $(p^{L}_i,\pi^{L}_i)\in(0,1]\times[0,1]$, with a value of zero if and only if $\pi^{L}_i=1/2$. If one probability $p^L_i$ is equal to zero, the $i$-th element of the sum is equal to zero, because of the convention used, namely $0\log_2(0)=0$. This leads to the conclusion of Theorem~\ref{th:positive}.
\end{proof}

\section{Proof of Proposition~\ref{pro:mgf_exact}}\label{sec:proofmgf_exact}

\begin{proof}
By definition
\begin{equation}\label{eq:MGF1}
M_{\widehat{I}^{L+1}}(t) = e^{tH^{L+1}_{\star}}\prod_{i=1}^{2^L}\E\left[\left.\exp\left(t\left[p^{L}_i\widehat{\pi}^{L}_i\log_2\left(p^{L}_i\widehat{\pi}^{L}_i\right) + p^{L}_i(1-\widehat{\pi}^{L}_i)\log_2\left(p^{L}_i(1-\widehat{\pi}^{L}_i)\right)\right]\right)\right|\pi^L_i=\frac{1}{2}\right].
\end{equation}
Under the condition $\pi^L_i=1/2$, the variable $n_i\widehat{\pi}^{L}_i$ follows a binomial distribution of parameters $n_i$ and $1/2$. We thus have $\proba(\widehat{\pi}^{L}_i=j/n_i)=\binom{n_i}{j}2^{-n_i}$ and:
\begin{equation}\label{eq:espMGF}
\begin{array}{cl}
 & \E\left[\left.\exp\left(t\left[p^{L}_i\widehat{\pi}^{L}_i\log_2\left(p^{L}_i\widehat{\pi}^{L}_i\right) + p^{L}_i(1-\widehat{\pi}^{L}_i)\log_2\left(p^{L}_i(1-\widehat{\pi}^{L}_i)\right)\right]\right)\right|\pi^L_i=\frac{1}{2}\right] \\
 = & \sum_{j=0}^{n_i} \exp\left(t\left[p^{L}_i\frac{j}{n_i}\log_2\left(p^{L}_i\frac{j}{n_i}\right) + p^{L}_i\left(1-\frac{j}{n_i}\right)\log_2\left(p^{L}_i\left(1-\frac{j}{n_i}\right)\right)\right]\right)\binom{n_i}{j}2^{-n_i} \\
 = & e^{tp^L_i\log_2\left(p^L_i\right)}\sum_{j=0}^{n_i} C^L_{i,j}(t).
 \end{array}
 \end{equation}
Combining this equation with equation~\eqref{eq:MGF1}, we get:
$$\begin{array}{ccl}
M_{\widehat{I}^{L+1}}(t) & = & e^{tH^{L+1}_{\star}}\prod_{i=1}^{2^L} e^{tp^L_i\log_2\left(p^L_i\right)} \sum_{j=0}^{n_i} C^L_{i,j}(t) \\
 & = & e^{tH^{L+1}_{\star}}e^{-tH^{L}}\prod_{i=1}^{2^L} \sum_{j=0}^{n_i} C^L_{i,j}(t) \\
  & = & e^{t}\prod_{i=1}^{2^L} \sum_{j=0}^{n_i} C^L_{i,j}(t),
\end{array}$$
because $H^{L+1}_{\star}=1+H^{L}$, according to equation~\eqref{eq:Hstar}.
\end{proof}

\section{Proof of Proposition~\ref{pro:moment}}\label{sec:proofmoment}

\begin{proof}
By definition, the $r$-th moment appearing in Proposition~\ref{pro:moment} is obtained by the $r$-th derivative of the moment-generating function, $M_{\widehat{I}^{L+1}}^{(r)}(0)$. By developing the product in the formula provided in Proposition~\ref{pro:mgf_exact}, we write
$$M_{\widehat{I}^{L+1}}(t)=e^{t}\sum_{j_1=0}^{n_1} \ldots\sum_{j_{2^L}=0}^{n_{2^L}}\prod_{i=1}^{2^L} C^L_{i,j_i}(t).$$
Then, Leibniz rule leads to
\begin{equation}\label{eq:proofmom1}
M_{\widehat{I}^{L+1}}^{(r)}(t)=\sum_{m=0}^r \binom{r}{m} e^{t}\sum_{j_1=0}^{n_1} \ldots\sum_{j_{2^L}=0}^{n_{2^L}}\left(\prod_{i=1}^{2^L} C^L_{i,j_i}(t)\right)^{(m)}.
\end{equation}
A straightforward calculation gives
\begin{equation}\label{eq:proofmom2}
\left(\prod_{i=1}^{2^L} C^L_{i,j_i}(t)\right)^{(m)}=\alpha_{j_1,...,j_{2^L}}^m\prod_{i=1}^{2^L} C^L_{i,j_i}(t).
\end{equation}
We also note that
\begin{equation}\label{eq:proofmom3}
C^L_{i,j_i}(0)=\binom{n_i}{j_i}\frac{1}{2^{n_i}}.
\end{equation}
Combining equations~\eqref{eq:proofmom1}, \eqref{eq:proofmom2}, and~\eqref{eq:proofmom3} together gives
$$\begin{array}{ccl}
M_{\widehat{I}^{L+1}}^{(r)}(0) & = & \sum_{m=0}^r \binom{r}{m} \sum_{j_1=0}^{n_1} \ldots\sum_{j_{2^L}=0}^{n_{2^L}}\alpha_{j_1,...,j_{2^L}}^m\prod_{i=1}^{2^L} C^L_{i,j_i}(0) \\
 & = & \sum_{m=0}^r \binom{r}{m} \sum_{j_1=0}^{n_1} \ldots\sum_{j_{2^L}=0}^{n_{2^L}}\alpha_{j_1,...,j_{2^L}}^m\prod_{i=1}^{2^L} \binom{n_i}{j_i}\frac{1}{2^{n_i}},
 \end{array}$$
which is the result stated in Proposition~\ref{pro:moment}.
\end{proof}

\section{Proof of Proposition~\ref{pro:deriv_gj}}\label{sec:proofderiv_gj}

\begin{proof}
Equation~\eqref{eq:deriv_gj} is a direct consequence of Faà di Bruno's formula to the function $x\mapsto \exp([itp^L_j/ln(2)]\Lambda(x))=g_j(x)$, where $\Lambda$ is a primitive of $\lambda$: $\Lambda(x)=x\ln(p^L_jx)+(1-x)\ln(p^L_j(1-x))-\ln(p^L_j/2)$. 

For proving equation~\eqref{eq:deriv_gj_1/2}, we simply observe by recurrence that, for $k\geq 1$:
\begin{equation}\label{eq:derlambda}
\frac{d^k}{dx^k}\lambda(x)=\frac{(k-1)!(-1)^{k+1}}{x^k}+\frac{(k-1)!}{(1-x)^k}.
\end{equation}
As a consequence, we have:
\begin{equation}\label{eq:recderlambda}
\frac{d^k}{dx^k}\lambda\left(\frac{1}{2}\right)=\left\{\begin{array}{ll}
0 & \text{if }k\in 2\mathbb N \\
(k-1)! 2^{k+1} & \text{if }k\in 2\mathbb N+1.
\end{array}\right.
\end{equation}
The Bell polynomial writes:
\begin{equation}\label{eq:Bell}
B_{k,l}\left(\lambda(x),\frac{d}{dx}\lambda(x),...\right) = \sum\frac{k!}{j_0!j_1!...j_{k-l}!}\prod_{m=0}^{k-l} \left(\frac{1}{(m+1)!}\frac{d^m}{dx^m}\lambda(x)\right)^{j_m},
\end{equation}
where the sum is taken over the set of $j_0,...,j_{k-l}$ submitted to the traditional Bell conditions, namely $\sum_m j_m=l$ and $\sum_m (m+1)j_m=k$. So, for $x=1/2$, we get:
$$\begin{array}{ccl}
B_{k,l}\left(\lambda\left(\frac{1}{2}\right),\frac{d}{dx}\lambda\left(\frac{1}{2}\right),...\right) & = & \sum\frac{k!}{j_0!j_1!...j_{k-l}!}\indic_{\{j_0=j_2=j_4=...=0\}}\prod_{m=1}^{k-l} \left(\frac{(m-1)!}{(m+1)!}2^{m+1}\right)^{j_m} \\
 & = & \sum\frac{k!}{j_0!j_1!...j_{k-l}!}\indic_{\{j_0=j_2=j_4=...=0\}}2^{\sum_m (m+1)j_m}\prod_{m=1}^{k-l} \left(\frac{(m-1)!}{(m+1)!}\right)^{j_m} \\
 & = & 2^kB_{k,l}(0,0!,0,2!,0,4!,0,...).
\end{array}$$
Combining this last equation and $g_j(t,1/2)=1$ in equation~\eqref{eq:deriv_gj}, we obtain equation~\eqref{eq:deriv_gj_1/2}. 

The reason why $\partial_x^k g_j(t,1/2)=0$ when $k$ is odd lies in the properties of Bell polynomial. Indeed, when we have $j_m=0$, for all $m$ even numbers, we have
$$\sum_{m=0}^{k-l} (m+1)j_m=\sum_{m=0}^{\lfloor\frac{k-l-1}{2}\rfloor} (2m+2)j_{2m+1}\in 2\mathbb N$$
and so this quantity cannot be equal to an odd $k$, as required by the second Bell condition. Therefore, at least one $j_m$ of even index $m$ must be different from 0 in each product of equation~\eqref{eq:Bell}. But, as the corresponding $m$-th derivative of $\lambda$ is equal to zero in $x=1/2$, the Bell polynomial (and the $k$-th derivative of $g_j$) is equal to zero.
\end{proof}

\section{Proof of Proposition~\ref{pro:average_deriv_gj}}\label{sec:proofaverage_deriv_gj}

In order to prove Proposition~\ref{pro:average_deriv_gj}, we first introduce and prove the following lemma.

\begin{lem}\label{lem:av1}
Let $x\in[1/2,1)$, $k\in\mathbb N$, the function $r$ be defined by equation~\eqref{eq:rk}, and $u_k(x)$ by $(1/(k+1)!)d^k\lambda(x)/dx^k$, where $\lambda$ is defined like in Proposition~\ref{pro:deriv_gj}. Then, for all $k'\in\llbracket 0,k\rrbracket$, we have $0\leq\min(u_0(x),u_1(x))\leq u_{k'}(x)\leq u_{r(k)}(x)$.
\end{lem}

\begin{proof}
The power series expansion of $x\mapsto u_k(x)$ leads to:
$$u_{k}(x)=\sum_{j=0}^{\infty}\frac{1}{j!(k+1)!}\frac{d^{k+j}\lambda(1/2)}{dx^{k+j}}\left(x-\frac{1}{2}\right)^j.$$
Noting thanks to equation~\eqref{eq:recderlambda} that $d^{k+2}\lambda(1/2)/dx^{k+2}=4k(k+1)d^{k}\lambda(1/2)/dx^{k}$, we also have:
$$\begin{array}{ccl}
u_{k+2}(x) & = & \sum_{j=0}^{\infty}\frac{1}{j!(k+3)!}\frac{d^{k+2+j}\lambda(1/2)}{dx^{k+2+j}}\left(x-\frac{1}{2}\right)^j \\
 & = & \sum_{j=0}^{\infty}\frac{4(k+j)(k+j+1)}{(k+2)(k+3)}\frac{1}{j!(k+1)!}\frac{d^{k+j}\lambda(1/2)}{dx^{k+j}}\left(x-\frac{1}{2}\right)^j.
\end{array}$$
Let $a_{k,j}=4(k+j)(k+j+1)/(k+2)(k+3)$. When $j\geq 2$, $a_{k,j}$ is trivially larger than 1. When $j=1$, we have $a_{k,j}=(2k+2)(2k+4)/(k+2)(k+3)>1$. Finally, when $j=0$, solving a simple binomial equation in $k$, we obtain that the only $k\in\mathbb N\setminus\{0\}$ such that $a_{k,0}<1$ is $k=1$. So, we have
\begin{equation}\label{eq:av_forall1}
\forall k\geq 2, u_{k+2}(x)\geq u_{k}(x).
\end{equation}
With a similar reasoning and solving a new binomial equation in $k$, we show that $\forall k\in\mathbb N\setminus\{0\}$, $16k(k+1)/(k+5)(k+4)>0$ and finally that 
\begin{equation}\label{eq:av_forall2}
\forall k\geq 1, u_{k+4}(x)\geq u_{k}(x).
\end{equation}
Moreover, if $k>0$ is an even number, we have from equation~\eqref{eq:derlambda} that $u_k(x)\leq 1/k(k+1)(1-x)^k$. The same equation also leads to $u_{k+1}(x)\geq 1/(k+1)(k+2)(1-x)^{k+1}\geq 2/(k+1)(k+2)(1-x)^{k}$. Noting that for $k\geq 2$, we have $2/(k+2)\geq 1/k$, this simply proves that
\begin{equation}\label{eq:av_forall3}
\forall k\geq 1, u_{2k}(x)\leq u_{2k+1}(x).
\end{equation}
Last, noting that $u_0(1/2)=0=u_2(1/2)$ and that
$$\frac{d}{dx}(u_0(x)-u_2(x))=\frac{1}{x}\left(1-\frac{1}{3x^2}\right)+\frac{1}{1-x}\left(1-\frac{1}{3(1-x)^2}\right) \leq \frac{1}{x}\frac{2}{3}-\frac{1}{1-x}\frac{4}{3} \leq -\frac{4}{3},$$
we get $u_0(x)\leq u_2(x)$. Combining this last equation with equations~\eqref{eq:av_forall1}, \eqref{eq:av_forall2}, and~\eqref{eq:av_forall3}, and noting that for $k\geq 5$, $r(k)=k$ (respectively $r(k)=k+1$) when $k$ is odd (resp. even), we prove Lemma~\ref{lem:av1}.
\end{proof}

We can now prove Proposition~\ref{pro:average_deriv_gj}.

\begin{proof}
Combining equation~\eqref{eq:deriv_gj} and Minkowski inequality, we get
\begin{equation}\label{eq:proofboundderiv1}
\left(\int_{1/2}^{z}{\left|\partial_x^k g_j(t,x)\right|^qdx}\right)^{1/q} \leq \sum_{l=1}^{k}\left(\int_{1/2}^{z}{\left|\frac{tp^L_j}{\ln(2)}\right|^{ql} \left|B_{k,l}\left(\lambda(x),\frac{d}{dx}\lambda(x),...,\frac{d^{k-l}}{dx^{k-l}}\lambda(x)\right)\right|^q dx}\right)^{1/q}.
\end{equation}
Using Lemma~\ref{lem:av1}, we can bound the Bell polynomial, in which we recall that $\sum_m j_m=l$ and $\sum_m (m+1)j_m=k$:
$$\begin{array}{ccl}
\left|B_{k,l}\left(\lambda(x),\frac{d}{dx}\lambda(x),...,\frac{d^{k-l}}{dx^{k-l}}\lambda(x)\right)\right| & = & \sum\frac{k!}{j_0!j_1!...j_{k-l}!}\prod_{m=0}^{k-l} \left(\frac{1}{(m+1)!}\frac{d^m}{dx^m}\lambda(x)\right)^{j_m} \\
 & \leq & \sum\frac{k!}{j_0!j_1!...j_{k-l}!} \left(\frac{1}{(r(k-l)+1)!}\frac{d^{r(k-l)}}{dx^{r(k-l)}}\lambda(x)\right)^{\sum_m j_m} \\
 & = & B_{k,l}\left(1!,2!,...,(k-l+1)!\right) \left(\frac{(r(k-l)-1)!}{(r(k-l)+1)!}\left(\frac{1}{x^{r(k-l)}}+\frac{1}{(1-x)^{r(k-l)}}\right)\right)^{l} \\
 & \leq & \mathcal L(k,l) \left(\frac{2}{(r(k-l)+1)r(k-l)}\right)^{l} (1-x)^{-lr(k-l)},
\end{array}$$
because $B_{k,l}\left(1!,2!,...,(k-l+1)!\right)=\mathcal L(k,l)$. We then use this bound in equation~\eqref{eq:proofboundderiv1}:
$$\begin{array}{ccl}
\left(\int_{1/2}^{z}{\left|\partial_x^k g_j(t,x)\right|^qdx}\right)^{1/q} & \leq & \sum_{l=1}^{k}\left(\left|\frac{tp^L_j}{\ln(2)}\right|^{ql} \mathcal L(k,l)^q \left(\frac{2}{(r(k-l)+1)r(k-l)}\right)^{ql} \int_{1/2}^{z}{(1-x)^{-qlr(k-l)}dx}\right)^{1/q} \\
 & = & \sum_{l=1}^{k}\left|\frac{tp^L_j}{\ln(2)}\right|^{l} \mathcal L(k,l) \left(\frac{2}{(r(k-l)+1)r(k-l)}\right)^{l} \left(\frac{(1-z)^{-qlr(k-l)+1}-2^{qlr(k-l)-1}}{qlr(k-l)-1}\right)^{1/q} \\
 & \leq & \sum_{l=1}^{k}\left|\frac{tp^L_j}{\ln(2)}\right|^{l} \mathcal L(k,l) \left(\frac{2}{(r(k-l)+1)r(k-l)}\right)^{l} \frac{(1-z)^{-lr(k-l)+1/q}}{(qlr(k-l)-1)^{1/q}}.
\end{array}$$
This concludes the proof.
\end{proof}

\section{Proof of Theorem~\ref{th:Taylor_gj}}\label{sec:Taylor_gj}

\begin{proof}
Noting that the derivatives are provided by Proposition~\ref{pro:deriv_gj} and more precisely by equation~\eqref{eq:deriv_gj_1/2} when $x=1/2$, we can do a Taylor expansion of $g_j$ around $x=1/2$: 
\begin{equation}\label{eq:preuveTaylor_gj}
g_j(t,x)=1+\frac{2itp^L_j}{\ln(2)}\left(x-\frac{1}{2}\right)^2+\mathcal R_j(t,x),
\end{equation}
where
$$\mathcal R_j(t,x) = \frac{1}{6}\int_{1/2}^x \partial_x^4g_j(t,y)(x-y)^3 dy.$$
Let $q>1$. According to Hölder's inequality, we have:
$$\begin{array}{ccl}
|\mathcal R_j(t,x)| & \leq & \frac{1}{6}\int_{1/2}^x \left|\partial_x^4g_j(t,y)(x-y)^3\right| dy \\
 & \leq & \frac{1}{6} \left(\int_{1/2}^x \left|\partial_x^4g_j(t,y)\right|^qdy\right)^{1/q} \left(\int_{1/2}^x \left|(x-y)^3\right|^{q/(q-1)}dy\right)^{1-1/q}.
\end{array}$$
The second integral is easily calculated and its value is $((q-1)/(4q-1))(x-1/2)^{(4q-1)/(q-1)}$. For the first integral, we use Proposition~\ref{pro:average_deriv_gj}, noting that the $r(k-l)$ are all equal to 5 since $k=4$. Therefore,
$$|\mathcal R_j(t,x)| \leq \frac{1}{6} \left(\sum_{l=1}^{4}\left|\frac{tp^L_j}{15\ln(2)}\right|^{l} \mathcal L(4,l) \frac{(1-x)^{-5l+1/q}}{(5ql-1)^{1/q}}\right)\left(\frac{q-1}{4q-1}\right)^{1-1/q}\left(x-\frac{1}{2}\right)^{4-1/q}.$$
Then, replacing $x$ by $X/n_j$ leads to
$$\E\left[\left|\mathcal R_j\left(t,\frac{X}{n_j}\right)\right|\right]  \leq  \frac{1}{6} \left(\frac{q-1}{4q-1}\right)^{1-1/q} \sum_{l=1}^{4}\left|\frac{tp^L_j}{15\ln(2)}\right|^{l} \mathcal L(4,l) \E\left[\frac{(1-X/n_j)^{-5l+1/q}}{(5ql-1)^{1/q}} \left(\frac{X}{n_j}-\frac{1}{2}\right)^{4-1/q}\right],$$
in which we can transform the expectation, which we note $\xi$, in the right-hand side of the inequality, thanks to the Cauchy-Schwarz inequality:
$$\xi^2 \leq \E\left[\frac{(1-X/n_j)^{-10l+2/q}}{(5ql-1)^{2/q}}\right] \E\left[\left(\frac{X}{n_j}-\frac{1}{2}\right)^{8-2/q}\right].$$
Noting that $1-X/n_j$ is distributed like $X/n_j$, we can apply an asymptotic result on the negative moments of a positive binomial variable~\cite[Corollary~3.4]{LG}: $\lim_{n_j\rightarrow\infty}\E[(1-X/n_j)^{-10l+2/q}]=2^{10l-2/q}$. Regarding the second equation, we first use Jensen's inequality, in which we impose to have $q\in\mathbb N$, and we conclude by noting that, for $d\in\mathbb N$, $\E[(X-n_j/2)^{2d}]\leq d^{2d}(n_j/4)^d$~\cite{Skorski}:
$$\begin{array}{ccl}
\E\left[\left(\frac{X}{n_j}-\frac{1}{2}\right)^{8-2/q}\right] & \leq & \left(\E\left[\left(\frac{X}{n_j}-\frac{1}{2}\right)^{8q-2}\right]\right)^{1/q} \\
 & \leq & (4q-1)^{8-2/q}(4n_j)^{-4+1/q}. 
\end{array}$$
As a consequence, for all $\varepsilon>1$, there exists $\nu\in\mathbb N$ such that for all $n_j\geq\nu$ and all $q$ and $l$, we have
$$\xi \leq \varepsilon \frac{2^{5l-4}(4q-1)^{4-1/q}}{(5ql-1)^{1/q}}n_j^{-2+1/2q}$$
and finally
\begin{equation}\label{eq:preuveTaylorReste}
\E\left[\left|\mathcal R_j\left(t,\frac{X}{n_j}\right)\right|\right]  \leq  \frac{\varepsilon}{96} \left(q-1\right)^{1-1/q} (4q-1)^{3}
\left(\sum_{l=1}^{4}\left|\frac{2^5tp^L_j}{15\ln(2)}\right|^{l} \frac{\mathcal L(4,l)}{(5ql-1)^{1/q}}\right) n_j^{-2+1/2q}.
\end{equation}
Going back to the Taylor expansion in equation~\eqref{eq:preuveTaylor_gj}, and knowing that the variance of the binomial variable $X$ is $n_j/4$, we then have:
$$\begin{array}{ccl}
\E\left[g_j\left(t,\frac{X}{n_j}\right)\right] & = & 1+\frac{2itp^L_j}{\ln(2)}\E\left[\left(\frac{X}{n_j}-\frac{1}{2}\right)^2\right]+\E\left[\mathcal R_j\left(t,\frac{X}{n_j}\right)\right] \\
 & = & 1+\frac{itp^L_j}{2\ln(2)n_j}+\E\left[\mathcal R_j\left(t,\frac{X}{n_j}\right)\right].
\end{array}$$
We define $R(t,n_j)=\E\left[\mathcal R_j\left(t,X/n_j\right)\right]$ and Jensen's inequality provides us with $|R(t,n_j)|\leq\E\left[\left|\mathcal R_j\left(t,X/n_j\right)\right|\right]$, for which we know an upper bound thanks to equation~\eqref{eq:preuveTaylorReste}. This concludes the proof.
\end{proof}

\end{document}